 \documentclass[conference]{IEEEtran}
\renewcommand{\baselinestretch}{0.95}
\IEEEoverridecommandlockouts
\usepackage[pdftex]{graphicx}
 \usepackage{lipsum,graphicx}
\usepackage{stackengine}

\usepackage{float}
\usepackage{amsmath,amssymb,amsfonts}
\usepackage{textcomp}
\usepackage{url}
\usepackage[T1]{fontenc}
\usepackage{mathtools, cuted}
\usepackage{microtype}
\usepackage{lipsum, color}
\usepackage{amsthm}
\newtheorem{theorem}{Theorem}

\newtheorem{lemma}{Lemma}

\usepackage{mathtools}
\usepackage{graphics}

\DeclarePairedDelimiter{\paren}{(}{)}
\DeclarePairedDelimiter{\abs}{\lvert}{\rvert}
\usepackage{algorithm}
\usepackage{algpseudocode}
\usepackage[nodisplayskipstretch]{setspace} 
\usepackage{float}
\usepackage{verbatim}
\usepackage{mathrsfs}
\usepackage{sansmath}
\usepackage{array}
\usepackage{cite}
\usepackage{lipsum}
\usepackage{mathtools}
\usepackage{icomma}
\usepackage{optidef}
\usepackage[font=footnotesize]{subcaption}
\usepackage[font=footnotesize]{caption}
\usepackage{mwe}
\usepackage{textgreek}
\usepackage{tabularx}
\usepackage{multirow}
\usepackage{stackengine}

\newcommand*{\Scale}[2][4]{\scalebox{#1}{$#2$}}

\ifCLASSINFOpdf
\else
\fi

\begin{document}


\title{Reconfigurable Intelligent Surface-Aided NOMA with Limited Feedback}
\author{
\IEEEauthorblockN{Mojtaba~Ahmadi~Almasi and Hamid~Jafarkhani}
\thanks{The authors are with Center for Pervasive Communications and Computing,
University of California, Irvine. This work was supported in part by the NSF
Award CNS-2008786.}
}

\maketitle

\vspace{-1cm}
\begin{abstract}
The design of feedback channels in frequency division duplex (FDD) systems is a major challenge because of the limited available feedback bits. We consider non-orthogonal multiple access (NOMA) systems that incorporate reconfigurable intelligent surfaces (RISs). In limited feedback RIS-aided NOMA systems, the RIS-aided channel and the direct channel gains should be quantized and fed back to the transmitter. This paper investigates the rate loss of the overall RIS-aided NOMA systems suffering from quantization errors. We first consider random vector quantization for the overall RIS-aided channel and  identical uniform quantizers for the direct channel gains. We then obtain an upper bound  for the rate loss, due to the quantization error, as a function of the number of  feedback bits and the size of RIS. Our numerical results indicate  the sum rate performance of the limited feedback system approaches that of the system with full CSI as the number of  feedback bits increases.  
\end{abstract}
\section{Introduction}\label{sec:introduction}
Reconfigurable intelligent surfaces (RISs) are presumed as an attractive solution to enhance the spectral, power efficiency, and coverage of wireless communication systems~\cite{8741198}. These surfaces consist of many passive and cost-effective elements capable of controlling the propagation environment by properly adjusting the direction of coming signals. These distinctive properties make RIS a promising solution for broad connectivity in the next generation of wireless networks.
Previously, intelligent surfaces that are not reconfigurable
~\cite{subrt2012intelligent, 8108330} and  
reconfigurable multiple-input multiple-output (MIMO) systems
~\cite{1367557,cetiner2008patent} have been proposed for orthogonal multiple access (OMA) systems. 
Also, it is shown that RISs can have notable use cases and boost performance when merged with other emerging technologies such as non-orthogonal multiple access (NOMA)~\cite{9000593,9079918}.   

NOMA has been a topic of research as a promising new technology for the next generation of wireless communications. Specifically, in the downlink, power-domain NOMA aims to serve two or more users by sharing the same time/frequency/code resource block~\cite{3GPP}. At the transmitter side, NOMA squeezes the information signals using superposition coding. Before decoding the intended signal, the stronger user applies successive interference cancellation, i.e., first decodes the weaker user's signal and then subtracts it from the received signal.  

In this paper, we incorporate the RIS in downlink NOMA to improve the quality of the  weak user's channel. Unlike many other time division duplex (TDD)-based RIS-assisted NOMA (RIS-NOMA) systems like~\cite{9000593, 9079918,9345507,9500202}, we consider a frequency division duplex (FDD) system. The FDD-based RIS-NOMA is more challenging in the sense that the channel must be estimated at the receiver and fed back to the transmitter via a limited feedback channel. The availability of the quantized channel gains instead of perfect channel state information (CSI) creates the following two major issues. First, the quantized channel gains can result in a severe rate loss~\cite{7968348,9103094}. This phenomenon is more harmful when the quantized channel gains inaccurately change the order of NOMA users. Second, the phase information obtained from the overall quantized channel restricts the performance of the RIS~\cite{9223720}. Motivated by this, we investigate the impact of quantizing the overall RIS-aided channel and the direct channel gains on the system's rate loss.     

The limited feedback problem in RIS-aided systems is studied in~\cite{chen2020adaptive,9347974, kim2021learning, 9625462}. Ref.~\cite{chen2020adaptive} proposes a cascaded codebook design and bit partitioning strategy in the presence of line-of-sight (LoS) and non-LoS (NLoS) channels. In~\cite{9347974}, the feedback is divided into two parts, channel feedback and angle information feedback. 
In particular,~\cite{9347974} uses a random vector quantizer (RVQ) codebook to quantize the channel followed by feeding back the indices related to the angle of arrival (AoA) and angle of departure (AoD) information of the cascade channel matrix. Similarly,~\cite{kim2021learning} designs a codebook-based limited feedback protocol for RIS using learning methods. In~\cite{9625462}, authors aim to reconstruct the channel using the signal strength feedback and exploiting the sparsity and low-rank properties. None of the above limited feedback methods  can be applied to our RIS-NOMA. Particularly, the feedback methods in~\cite{9347974,chen2020adaptive, kim2021learning} mainly try to send the normalized channel vectors to perform beamforming at the transmitter. Further,~\cite{9625462} determines the channel direction while the channel gains are estimated based on the distance and not precisely. 
However, our RIS-NOMA system requires precise channel gains to accurately order the users and perform the superposition coding. Also, it needs the overall RIS-aided channel vector to adjust the RIS. 
The main contributions of this paper are as follows: 
\begin{itemize}
    \item We provide a limited feedback framework for RIS-NOMA systems and analyze its performance.
    \item We find an upper bound on the rate loss caused by quantization.
\end{itemize}
We conduct numerical simulations to evaluate the sum rate and the rate loss of the limited feedback RIS-NOMA system. The results verify our theoretical derivations.      


\textbf{Notations:} In this paper, $j=\sqrt{-1}$. Small letters, bold letters, and bold capital letters designate
scalars, vectors, and matrices, respectively. Superscripts $(\cdot)^T$ and $(\cdot)^\dagger$ denote, respectively, the transpose and the transpose-conjugate operations.  Further, $|x|$, $\mathbb{E}[x]$, and $\mathbb{V}[x]$ denote the absolute, the expected, and the variance of $x$, respectively. The operation $\angle \mathbf{x}$  calculates the element-wise angle of the vector $\mathbf{x}$ and $\lfloor\cdot\rfloor$ is the floor function. Finally, $\gamma(n, x) = \int_0^x t^{n-1}e^{-t}dt$ denotes the lower incomplete gamma function.   

\section{System Model}\label{sec:system_model}
Our system model is shown in Fig.~\ref{fig_sum_rate} in which a single-input single-output system model similar to~\cite{9333565,9693963} is considered. NOMA is a suitable multiple access technique for a single-antenna setup because multi-user solutions are not applicable.
The base station (BS) uses NOMA to simultaneously serve two users named User~1 and User~2.\footnote{User pairing is out of the scope of this paper. We assume that the users are paired using one of the existing methods in the literature such as~\cite{8016604,kim2021lowcomplexity,9500202}. The complexity of rate loss calculation increases as the number of NOMA users grows.}
The distance from User~2 to the BS is more than that of User~1. To determine which user is near, the BS captures the distance information through a channel quality indicator (CQI). 
The purpose of the RIS is to serve the far user to improve the channel quality\cite{9000593,9079918}. The RIS is equipped with $N$ antenna elements that ideally can direct the incident signal to any arbitrary directions in $[-\frac{\pi}{2}, \frac{\pi}{2}]$.  

Like~\cite{9024675,9139273}, it is assumed that the perfect CSI is estimated at the users. This assumption allows us to focus on studying the impact of the channel gain and phase vector quantization errors. Recently,~\cite{omid2022low} has investigated the impact of CSI impairments such as erroneous channel estimation or delay in feedback on beamforming in a RIS-NOMA system without considering quantization error. 
As a promising solution, our work on beamforming in relay networks with channel statistics \cite{jing2008network} and quantized feedback~\cite{4641953,7745957} can be extended to the underlying  RIS-NOMA system to study the CSI impairments. 

\begin{figure}[!t]
    \centering
    \includegraphics[scale=0.8]{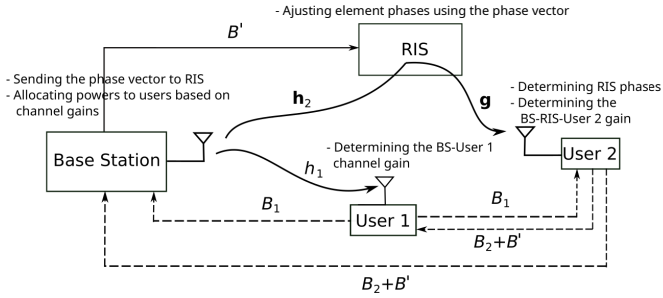}
    \caption{The limited feedback RIS-NOMA system model. $B_1$ bits and $B_2+B^\prime$ bits are allocated to User~1 and  User~2, respectively.}
    \label{fig_sum_rate}
\vspace*{-0.6cm}
\end{figure} 

\subsection{Transmit Channel Model}
We recall that both users are capable of estimating their channels. That is, User 1 obtains $h_1\in \mathbb{C}$
and User 2 obtains $\mathbf{h}_2\in\mathbb{C}^{N\times 1}$ and $\mathbf{g}\in\mathbb{C}^{N\times 1}$~\cite{8879620}. The channels capture the small-scale fading and path loss effects. For User~1, $h_1 = \sqrt{{L}_1}h_1^\prime$, where $L_1 = d_1^{-\alpha_1}$ is due to the path loss. The parameters $d_1$ and $\alpha_1$ denote the distance between the BS and User~1 and the path loss factor, respectively. Further, $h_1^\prime\sim \mathcal{CN}(0, 1)$ denotes the small-scale Rayleigh fading. We define the channel gain $H_1 = \abs*{h_1}^2$ with the probability density function (pdf) of  $f_{H_1}\paren{H_1} = \frac{1}{L_1}e^{-\frac{H_1}{L_1}}$.
Also, the channel between the BS and the RIS is defined as $\mathbf{h}_2 = \sqrt{L_2}\mathbf{h}_2^\prime$, where $L_2 = d_2^{-\alpha_2}$ and $\mathbf{h}_2^\prime\sim \mathcal{CN}(\boldsymbol{0}, \mathbf{I})$. The channel between the RIS and User~2 is given by $\mathbf{g} = \sqrt{L_g}\mathbf{g}^\prime$, where $L_g = d_g^{-\alpha_g}$ and $\mathbf{g}^\prime \sim \mathcal{CN}(\boldsymbol{0}, \mathbf{I})$. We note that $\mathbf{I}$ represents the identity matrix of size $N\times N$. In fact, the small-scale fading from the BS to User~2 is subject to the double-Rayleigh fading~\cite{9333565, 9478912, trigui2021bit}. Another possible channel model from the BS to the RIS is Rician fading. Since the RIS is helping the far user, it is reasonable to assume that its distance from the BS is large~\cite{9000593,9079918}. Thus, it is likely that the LoS channel is blocked by moving objects or buildings justifying a Rayleigh fading model. The parameters $d_2$ and $d_g$ denote the distance between the transmitter and the RIS and the distance from the RIS to User~2, respectively. Further, the parameters $\alpha_2$ and $\alpha_g$ denote the path loss factors. The effective overall channel between the BS and User 2 is defined as $\tilde{h}_2 = {\mathbf{g}^T\boldsymbol{\Theta}\mathbf{h}_2}$. Correspondingly, the channel gain is $\tilde{H}_2 = \abs*{\mathbf{g}^T\boldsymbol{\Theta}\mathbf{h}_2}^2$ in which $\boldsymbol{\Theta} = \text{diag}\paren*{\boldsymbol{\theta}}$, where $\boldsymbol{\theta} = [e^{j\phi_1}, \cdots, e^{j\phi_N}]$ and $\phi_i \in [-\frac{\pi}{2}, \frac{\pi}{2}]$ reflect the impact of the RIS.  
The optimal values of $\boldsymbol{\Theta}$ result in the maximum gain of 
$H_2 = \paren*{\sum_{i=1}^N\abs*{h_{2,i}}\abs*{g_i}}^2$. Deriving the exact pdf of $H_2$ is complicated. For the sake of simplicity, we first use the following upper bound on the pdf of the random variable $z = \sum_{i=1}^N|h_{2,i}||g_i|$~\cite{9079918}:
\begin{equation} \label{eq_PDF_z}
    \Scale[0.9]{f_z(z) \leq \frac{C_1}{C_2} \paren*{\frac{z}{C_2}}^{\frac{3N-2}{2}}e^{-2\frac{z}{C_2}},} 
\end{equation}
where $C_1 = \frac{2^N\pi^\frac{N}{2}\Gamma^N\paren*{\frac{3}{2}}}{\paren*{\frac{3N-2}{2}}!}$ and $C_2 = \sqrt{L_2L_g}$. Through extensive simulations, it is shown that this bound is tight~\cite{9079918}. Without loss of generality, we assume $N$ is an even number. Then, noting that $H_2 = z^2$, we have $f_{H_2}\paren*{H_2} = \frac{1}{2\sqrt{H_2}}f_z\paren*{\sqrt{H_2}}$. Finally, an upper bound on the pdf of $H_2$ follows as:
\begin{equation}\label{eq_PDF_H2}
    \Scale[0.9]{f_{H_2}\paren*{H_2} \leq \frac{C_1}{C_2}\frac{1}{2\sqrt{H_2}}\paren*{\frac{\sqrt{H_2}}{C_2}}^{\frac{3N-2}{2}}e^{-\frac{2\sqrt{H_2}}{C_2}}.}
\end{equation}

\subsection{Feedback Channel}
In FDD systems, the downlink channel is estimated at the user side and then fed back to the BS and the other user using the limited feedback channel. 
In our system model shown in Fig.~\ref{fig_sum_rate}, User~1 quantizes the channel gain $H_1$ and maps it to $q(H_1)$. Then, the index of $q(H_1)$ is fed back to the BS using $B_1$ bits. Since User 2 communicates with the BS through the RIS, the phase information should be sent to the RIS as well. In this regard, first, User~2 maps the overall channel vector $\mathbf{G}\mathbf{h}_2$ to $Q(\mathbf{G}\mathbf{h}_2)$ using  ${B}^\prime$ bits, where $Q(\cdot)$ is a RVQ and $\mathbf{G} = \text{diag}(\mathbf{g})$. 
User~2 then determines the phase vector $\boldsymbol{\theta}$ of the quantized channel $\mathbf{G}\mathbf{h}_2$ denoted by $\boldsymbol{\theta}_Q$. The exact structure of quantizers $q(\cdot)$ and $Q(\cdot)$ is discussed in the next section, but does not change the overall characteristics of the feedback channel model.
Next, User~2 quantizes the channel gain $H_{2,Q} = |\boldsymbol{\theta}_Q^T\mathbf{G}\mathbf{h}_2|^2$, i.e., $q(H_{2,Q})$, with $B_2$ bits.  
User~2 feeds the corresponding ${B_2} + {B}^\prime$ bits back to the BS. The feedback link from User~2 to the BS is assumed to support $B_2+B^\prime$ bits, although, the direct links might be blocked~\cite{chen2020adaptive}. The same explanation holds for the feedback link from User~1 to the BS.  

In our system model, the BS, the RIS, and the users  are fixed and the phase vector $\boldsymbol{\theta}_Q$ and the channel gains $H_1$ and $H_2$ are only required once for every channel coherence time.   

\subsection{Sum Rate}
To maximize the sum rate by efficiently allocating the transmit power and subject to some minimum rate for each user, the following optimization problem can be defined
\begin{maxi!}
{{\beta}}{R_1 + R_2\label{eq:objective_power_simp}}{\label{eq:opt_power_simp}}{}
\addConstraint{{R_1, R_2}}{{\geq R_{th}}\label{power_a}}
\addConstraint{{P_1 + P_2}}{{= P,}\label{power_b}}
\end{maxi!}
where $\Scale[0.9]{R_1 = \text{log}_2\paren*{1 + \beta P H_1}}$ and $\Scale[0.9]{R_2 = \text{log}_2\paren*{1 + \frac{\paren*{1-\beta}PH_2}{\beta P H_2 + 1}}}$. The power allocation can be parameterized by a factor $\beta$ such that $P_1 = \beta P$ and $P_2 = \paren*{1-\beta}P$. In~\cite{7890454}, given full CSI,  Problem~\eqref{eq:opt_power_simp} is extended to an arbitrary number of users and individual minimum rate constrains. Using the solution given in~\cite[Eq.~15]{7890454}, we obtain the optimum factor $\beta$ for $H_2 \leq H_1$ as $\Scale[0.9]{\beta^* = \frac{PH_2 - \epsilon}{\paren*{1 + \epsilon}P H_2}}$ where $\epsilon = 2^{R_{th}}-1$. Hence, $\Scale[0.9]{R_1 = \text{log}_2\paren*{1 + \frac{PH_2H_1 - \epsilon H_1}{H_2\paren*{1 + \epsilon}}}}$ and $R_2 = R_{th}$. The threshold $R_\textit{th}$ is the same for all users to ease the formulation but the approach works for arbitrary thresholds. Further, there are other useful objective functions to be considered. 
For example, similar to Ref.~\cite{7968348}, we can study the rate fairness in our RIS-NOMA system with limited feedback and quantization error. 

\section{Uniform and Random vector Quantizers}\label{sec:quantizer}
In this section, we describe uniform quantizers and RVQs, used in our system. We use uniform quantizers to compress the scalar channel gains  and RVQs to quantize the overall channel vector. 


To quantize $H_1$, we define the following uniform quantizer $q: \mathbb{R} \to \mathbb{R}$, shown in Fig.~\ref{fig:quantizer}: 
\begin{align}\label{eq:quantizer}
    q(x) = \begin{cases} \lfloor \frac{x}{\delta}\rfloor \times \delta, & x \leq \paren*{2^{B_1}-1}\delta, \\
    \paren*{2^{B_1}-1}\delta, & x > \paren*{2^{B_1}-1}\delta,
    \end{cases}
\end{align}
where $x$ is any non-negative real number and $\delta$ denotes the size of quantization partitions. The index of $q(H_1)$ is fed back by $B_1$ bits.  
The method in~\eqref{eq:quantizer} quantizes the channel gain to the left boundary of the partition instead of the center point. When the gain is quantized to the center point, the quantized value might be higher than the true channel gain. This can frequently cause outage at the weak user, i.e., solving the optimization problem in~\eqref{eq:opt_power_simp} may result in allocating insufficient power to the weak user. Thus,  Constraint~\eqref{power_a} may not hold for the quantized channel gain. To avoid the outage, we quantize the channel gain to the left boundary which guarantees the power allocation to the weak user is more than the needed optimal value.\footnote{The optimal power allocation, i.e., $P_1^*$ and $P_2^*$, is obtained by determining $\beta^*$ in~\eqref{eq:opt_power_simp}.} The uniform quantization is selected for simplicity but our approach works for non-uniform quantization as well.

\begin{figure}
    \centering
    \includegraphics[scale=0.8]{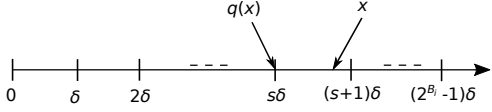}
    \caption{Applied uniform quantizer for quantizing channel gains ($i=1, 2$).}
    \label{fig:quantizer}
    \vspace*{-0.6cm}
\end{figure}

In general, there are two options for quantizing  the vector $\mathbf{G}\mathbf{h}_2$: vector quantization and scalar quantization, applied to vector's elements. 
Since the number  of elements in RISs can be large, scalar quantization will require a huge number of feedback bits and may not be practical. Inspired by this, we use random vector quantization in which the feedback bits can be far less than the number of elements\footnote{RVQ is a simple method but not the most efficient to quantize a vector with a limited number of feedback bits~\cite{6272722}. 
Other techniques such as
Lloyd Algorithm~\cite{gersho2012vector} 
and variable-length limited feedback beamforming~\cite{6897978} which outperform the RVQ can be used to enhance the performance of the underlying limited feedback system although they increase the complexity.}. We define the RVQ codebook $\mathcal{W} = \{\mathbf{w}_1, \mathbf{w}_2, \dots, \mathbf{w}_M\}$, in which the codeword $\mathbf{w}_i \in \mathbb{C}^{N\times 1}$, is  the quantized overall RIS-aided channel vector  $\mathbf{G}\mathbf{h}_2$. The codebook $\mathcal{W}$ is generated by selecting each of $M=2^{B^\prime}$ vectors independently from a uniform distribution on the complex unit sphere~\cite{1715541}. 
 
We aim to maximize the channel gain using the codebook such that
\begin{equation}\label{eq_max_w}
    \Scale[0.9]{Q\paren*{\mathbf{G}\mathbf{h}_2} = \underset{\mathbf{w}\in \mathcal{W}}{\text{argmax}} \abs{\mathbf{w}^\dagger\mathbf{G}\mathbf{h}_2}^2.}
\end{equation}
Further, we let $\Scale[0.9]{\boldsymbol{\phi}_Q = \angle Q(\mathbf{G}\mathbf{h}_2)}$ and $\Scale[0.9]{\boldsymbol{\theta}_Q = e^{j\boldsymbol{\phi}_Q}}$ where $\Scale[0.9]{\boldsymbol{\phi}_Q = [\phi_{1,Q}, \cdots, \phi_{N,Q}]^T}$. The channel gain $\Scale[0.9]{H_{2,Q} = \abs*{\boldsymbol{\theta}_Q^T\mathbf{G}\mathbf{h}_2}^2}$ takes any non-negative value 
and is mapped  to $q\paren*{H_{2,Q}}$ using the quantizer in~\eqref{eq:quantizer}. As mentioned before, a uniform quantizer is applied to $H_{2,Q}$ instead of $H_2$. Since $H_2$ and $H_{2,Q}$ do not necessarily belong to the same partition, their quantized values might be different, i.e., $q(H_{2, Q}) \leq q(H_2)$. The index of $q\paren*{H_{2,Q}}$ is sent using $B_2$ feedback bits. The uniform quantizers applied to $H_1$ and $H_{2,Q}$ include the same number of partitions, i.e., $B_1=B_2=B$.     
Defining $ \eta=\frac{H_{2,Q}}{H_2}$, we have   $\eta \in [0, 1]$.   
Deriving the pdf of $\eta$ is not straightforward, but needed in some analysis later. The exact pdf of $\eta$ for a $2\times 1$ Rayleigh channel vector and a large number of feedback bits is derived in ~\cite{4034215}. 
However, each element in $\mathbf{Gh}_2$ is a double-Rayleigh variable and its pdf is different from that of Rayleigh channels. Recently, the study of the pdf of $\Scale[0.9]{\Big|\sum_{i=1}^{N}|h_{2,i}||g_i|e^{j\kappa_i}\Big|^2}$ has been the topic of research in RIS-aided systems~\cite{9138463,9333565}. Note that $\kappa_i$ does not necessarily equal to $\phi_{i,Q}$. In~\cite[Lemma 1]{9333565}, the pdf of the random variable $\Scale[0.9]{\Big|\sum_{i=1}^{N}|h_{2,i}||g_i|e^{j\kappa_i}\Big|^2}$ where $\kappa_i$ is treated as a phase-noise is accurately approximated by a Gamma random variable. Following the same approach, our extensive empirical study shows that the pdf of the random variable $\eta$ can be approximated by the pdf of a beta random variable with the shape parameters $r_1$ and $r_2$, i.e.,
\begin{equation}\label{eq:eta_beta}
    \Scale[0.9]{f_{\eta}(\eta) \approx \frac{1}{\mathcal{B}\paren{r_1,r_2}} \eta^{r_1-1}\paren*{1-\eta}^{r_2-1},}
\end{equation}
where  $ \Scale[0.9]{r_1=\paren*{\frac{\mathbb{E}[\eta]\paren*{1-\mathbb{E}[\eta]}}{\mathbb{V}[\eta]}}\mathbb{E}[\eta]}$, $ \Scale[0.9]{r_2 = \paren*{\frac{\mathbb{E}[\eta]\paren*{1-\mathbb{E}[\eta]}}{\mathbb{V}[\eta]}}\paren*{1-\mathbb{E}[\eta]}}$, and $\mathcal{B}\paren*{r_1,r_2}=\frac{\Gamma(r_1+r_2)}{\Gamma(r_1)\Gamma(r_2)}$ represents a normalization constant that ensures the total probability is 1. The values of $\mathbb{E}[\eta]$ and $\mathbb{V}[\eta]$ depend on the RIS's size $N$ and the feedback bits $B^\prime$. The empirical results show a close resemblance between the real pdf and the approximation, but we do not have space to show them in this paper. 

When the full CSI is available at the BS, the user ordering is always accurate and the rates are calculated using~\eqref{eq:opt_power_simp}. In a limited feedback system,
an inaccurate user ordering can impose severe rate loss. Even if the user ordering is accurate, the quantization can reduce the achievable sum rate. 

\section{Rate Loss Analysis}\label{sec:rate_loss}
Let us name the user with a higher channel gain  the \textit{strong user}. We calculate the rate loss only for the strong user because the weak user will not experience quantization rate loss. When instead of full CSI, the quantized channel gains are used in Problem~\eqref{eq:opt_power_simp}, we call the resulting power allocation factor $\beta_q$.  

In general, the strong user's rate loss is obtained as
\begin{align}\label{eq_rateloss}
    \Scale[0.9]{\Delta R} & \Scale[0.9]{= R_i - R_{i,q}   = \text{log}_2\paren*{1 + P X_i} - \text{log}_2\paren*{1 + P X_{i,q}}} \nonumber \\
    & \Scale[0.9]{= \text{log}_2\paren*{1 + \frac{P \Delta X}{1 + PX_{i,q}}}\leq \text{log}_2\paren*{1 + P\Delta X},}
\end{align}
where $\Delta X = X_i - X_{i,q}$ denotes the normalized signal-to-noise ratio (SNR) loss. When user ordering is accurate, we have $\Scale[0.9]{X_1 = \beta H_1~(X_2 = \beta H_2)}$ and $\Scale[0.9]{X_{1,q} = \beta_q H_1~(X_{2,q} = \beta_q H_{2,Q})}$. For an inaccurate user ordering, $X_i$ is similar to the accurate one while $\Scale[0.9]{X_{1,q} = \frac{\paren*{1-\beta_q}Pq\paren*{H_1}}{\beta_q P q\paren*{H_1}+1}}$ and $\Scale[0.9]{X_{2,q} = \frac{\paren*{1-\beta_q}Pq\paren*{H_{2,Q}}}{\beta_q P q\paren*{H_{2,Q}}+1}}$. Based on~\eqref{eq_rateloss}, an upper bound on the average rate loss is found as
\begin{align}\label{eq_delta_x}
    \mathbb{E}[\Delta R] & \leq \mathbb{E}[\text{log}_2\paren*{1 + P\Delta X}] 
    \leq \text{log}_2\paren*{1 + P \mathbb{E}[\Delta X]}.
\end{align}
The second inequality is due to Jensen's inequality. Thus, to find the upper bound, we first derive $\Delta X$ and then calculate the expectation of $\Delta X$, i.e., $\mathbb{E}[\Delta X]$. However, the calculation of $\Delta X$ heavily depends on the values of $H_1$ and $H_2$. The main three Super Regions for this calculation, as shown in Fig.~\ref{fig:graphic},  are:
\begin{itemize}
    \item Super Region I: This consists of the conditions in which $q(H_1)=0$ and/or $q(H_{2,Q})=0$ which results in $\beta_q= \infty$. 
    \item Super Region II: This includes the main partitions of User~2's uniform quantizer, i.e., $\delta \leq q(H_{2,Q}) < \paren*{2^B-1}\delta$. 
    \item Super Region III: This includes the upper marginal partition of User 2's uniform quantizer, i.e., $q(H_{2,Q}) = \paren*{2^B-1}\delta$.
\end{itemize}
We denote $\Delta X$ in Super Region I as $\Delta X_\text{I}$. In what follows, we  calculate $\mathbb{E}[\Delta X_\text{I}]$, i.e., the expected  normalized SNR loss in each region. 
\begin{figure}
    \centering
    \includegraphics[scale=0.45]{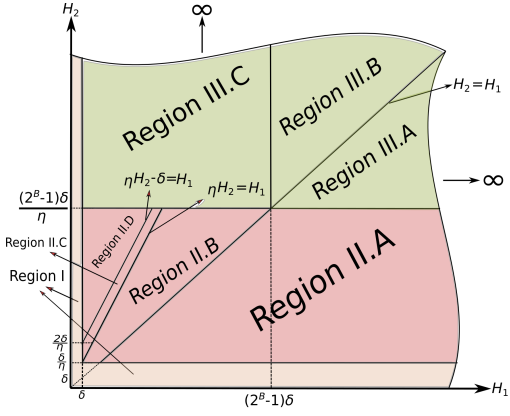}
    \caption{Presentation of all the possible regions for the rate loss.}
    \label{fig:graphic}
    \vspace*{-0.6cm}
\end{figure}

\subsection{Super Region I} Since $\beta_q= \infty$,
 in this super region, NOMA is not feasible and we let $\mathbb{E}[\Delta X_\text{I}]=0$.  

\subsection{Super Region II} 
\begin{lemma} \label{lemma_1}
\normalfont
The total average normalized SNR loss of Super Region~II is 
\begin{align}\label{eq:DXII}
    \Scale[0.83]{\mathbb{E}[\Delta X_\text{II}] \leq \sqrt{\paren*{2^B-1}\delta}\Bigg[C_3
      +\delta\paren*{C_4 + C_5\text{E}_1\paren*{\frac{\delta}{L_1}}}\Bigg]  + C_5 \sqrt{2\delta},}
\end{align}
where $\Scale[0.9]{C_3 = C_6+C_7 + C_8}$ and $\Scale[0.9]{C_4 = C_9 + C_{10}}$. In detail, $\Scale[0.9]{C_5 \geq C_5^{'}\mathbb{E}\left[ \frac{1}{\sqrt{\eta}}\right]}$, $\Scale[0.9]{C_6 \geq C_6^{'}  \mathbb{E}\left[ \frac{1-\eta}{\sqrt{\eta}}\right]}$, $\Scale[0.9]{C_7 \geq C_7^{'}  \mathbb{E}\left[ \frac{1-\eta}{\sqrt{\eta}}\right]}$, $\Scale[0.9]{C_8 \geq C_7^{'}\mathbb{E}\left[(1-\eta)\sqrt{\eta}\right]}$, $\Scale[0.9]{C_9 \geq C_9^{'}  \mathbb{E}\left[ \frac{1}{\sqrt{\eta}}\right]}$, and $\Scale[0.9]{C_{10} \geq C_{10}^{'}\mathbb{E}\left[ \frac{1}{\sqrt{\eta}}\right]}$ in which $\Scale[0.9]{C_5^{'} =\frac{C_1C_2}{2^{\frac{3N+2}{2}}}\paren*{\frac{3N+2}{2}}!}$, $\Scale[0.9]{C_6^{'} = \frac{C_1C_2}{2^{\frac{3N+2}{2}}}\paren*{\frac{3N+2}{2}}! + \frac{C_1L_1}{2^{\frac{3N-2}{2}}C_2}\paren*{\frac{3N-2}{2}}!}$, $\Scale[0.9]{C_7^{'} = \frac{C_1C_2^3}{2^{\frac{3N+6}{2}}L_1}\paren*{\frac{3N+6}{2}}!}$, $\Scale[0.9]{C_9^{'} = \frac{C_1}{2^{\frac{3N-2}{2}}C_2}\paren*{\frac{3N-2}{2}}! + \frac{C_1L_1}{2^{\frac{3N-6}{2}}C_2^3}\paren*{\frac{3N-6}{2}}!}$, and $\Scale[0.9]{C_{10}^{'} = \frac{C_1C_2}{2^{\frac{3N+2}{2}}L_1}\paren*{\frac{3N+2}{2}}!}$. Also, $ \Scale[0.9]{\text{E}_1(x) = \int_x^\infty \frac{e^{-t}}{t}dt}$ denotes the exponential-integral function.  
\end{lemma}
\begin{proof}
Please see Appendix~\ref{Appendix_A}.
\end{proof}
\subsection{Super Region III} 
\begin{lemma} \label{lemma_2}
\normalfont    
The total average normalized SNR loss for Super Region~III is bounded by
\begin{align}\label{eq:DXIII}
    &\Scale[0.85]{\mathbb{E}[\Delta X_\text{III}]  \leq e^{-2\frac{\sqrt{\paren*{2^B-1}\delta}}{C_2}}\Bigg[2C_{11}\paren*{1 + \frac{\paren*{2^B-1}\delta}{L_1}}}\times \nonumber \\
    &\Scale[0.85]{\paren*{1+\paren*{ 2\sqrt{\frac{\paren*{2^B-1}\delta}{  C_2^2}}}^{\frac{3N+2}{2}}}  
    + C_{12}\paren*{1 + \paren*{2\sqrt{\frac{\paren*{2^B-1}\delta}{ C_2^2}}}^{\frac{3N-2}{2}}}\Bigg]}.
\end{align}
where $C_{11} = \frac{C_1C_2^2}{2^{\frac{3N+4}{2}}}\paren*{\frac{3N+4}{2}}!$ and $C_{12} = \frac{C_1C_2L_1}{2^\frac{3N+2}{2}}\paren*{\frac{3N}{2}}!$. 
\end{lemma}
\begin{proof}
    Please see Appendix~\ref{Appendix_C}.
\end{proof}

Finally, we have the following theorem on the expectation of the total rate loss for the quantized RIS-NOMA.  
\begin{theorem} \label{theorem_1}
\normalfont
The total average rate loss for the quantized RIS-NOMA system with limited feedback is upper bounded by
\begin{equation}\label{eq:rate_loss_upper_bound}
\mathbb{E}[\Delta R] \leq \text{log}_2\paren{1 + P\mathbb{E}[\Delta X]},   
\end{equation}
where 
\begin{align}\label{eq:DX}
     \Scale[0.8]{\mathbb{E}[\Delta X]} &  \Scale[0.8]{\leq  \sqrt{\paren*{2^B-1}\delta}\Bigg[C_3
      +\delta\paren*{C_4 + C_5\text{E}_1\paren*{\frac{\delta}{L_1}}}\Bigg]  + C_5 \sqrt{2\delta}} \nonumber \\  
      & \quad  \Scale[0.8]{+ e^{-2\frac{\sqrt{\paren*{2^B-1}\delta}}{C_2}}\Bigg[ 2C_{11}\paren*{1 + \frac{\paren*{2^B-1}\delta}{L_1}}  \paren*{1+\paren*{ 2\sqrt{\frac{\paren*{2^B-1}\delta}{ C_2^2}}}^{\frac{3N+2}{2}}}}\nonumber \\
    & \quad \quad \quad \quad \quad  \quad\quad \Scale[0.85]{+ C_{12}\paren*{1 + \paren*{2\sqrt{\frac{\paren*{2^B-1}\delta}{ C_2^2}}}^{\frac{3N-2}{2}}}\Bigg]}.
\end{align}
\end{theorem}
\begin{proof}
We know that $\mathbb{E}[\Delta X] = \mathbb{E} [\Delta X_\text{I}]+ \mathbb{E} [\Delta X_\text{II}]+\mathbb{E}[\Delta X_\text{III}]$. Noting $\mathbb{E} [\Delta X_\text{I}]=0$ and replacing $\mathbb{E}[\Delta X_\text{II}]$ and $\mathbb{E}[\Delta X_\text{III}]$ with~\eqref{eq:DXII} and~\eqref{eq:DXIII}, respectively, results in~\eqref{eq:DX}.
\end{proof}
To guarantee that the rate loss approaches  zero as $B$ increases, one feasible solution is to define $\delta = \zeta_1 2^{-{\zeta_2} B}$ for $\zeta_1, \zeta_2 \in (0, 1)$. The parameters $\zeta_1$  and $\zeta_2$ are design parameters and should be optimized. Such a parameter  optimization is out of the scope of this paper although the choice of $\zeta_1$ and $\zeta_2$ will not affect the system model. In simulations, we will intuitively select $\zeta_1$ and $\zeta_2$ to achieve good performance. 

\section{Numerical Results}\label{sec:numerical_results}
We compare the sum rate performance of the RIS-NOMA system with limited feedback to that of the RIS-assisted orthogonal multiple access (RIS-OMA) system. The RIS-NOMA system is described in Section~\ref{sec:system_model}. For the RIS-OMA, we consider the same system model in Section~\ref{sec:system_model} and replace NOMA with OMA. 
Furthermore, since power allocation is not required in the RIS-OMA, the channel gains are not fed back although the RIS-aided channel vector information should be fed back for beamforming at the RIS. 

The parameters are set according to~\cite{9345507} as follows. The distances are selected as $d_1$=10 m, $d_2$=40 m, and $d_g$=10 m. Further, the path loss exponents are set as $\alpha_1 = 3.5$, $\alpha_2 = 2.5$, and $\alpha_g = 2.5$. The number of RIS elements is set to $N=10$. 

\begin{figure}
     \centering
    \begin{subfigure}[t]{0.24\textwidth}
        \raisebox{-\height}{\includegraphics[scale=0.43]{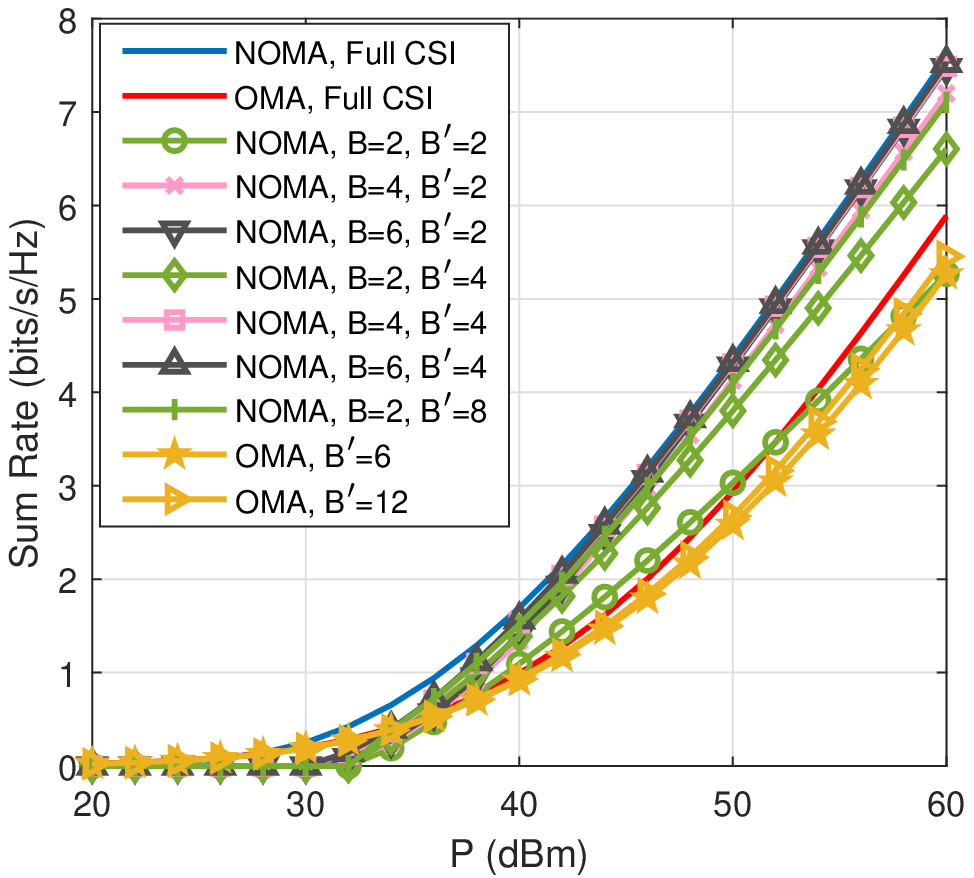}}
        \caption{}
    \end{subfigure}
    \hfill
    \begin{subfigure}[t]{0.24\textwidth}
        \raisebox{-\height}{\includegraphics[scale=0.43]{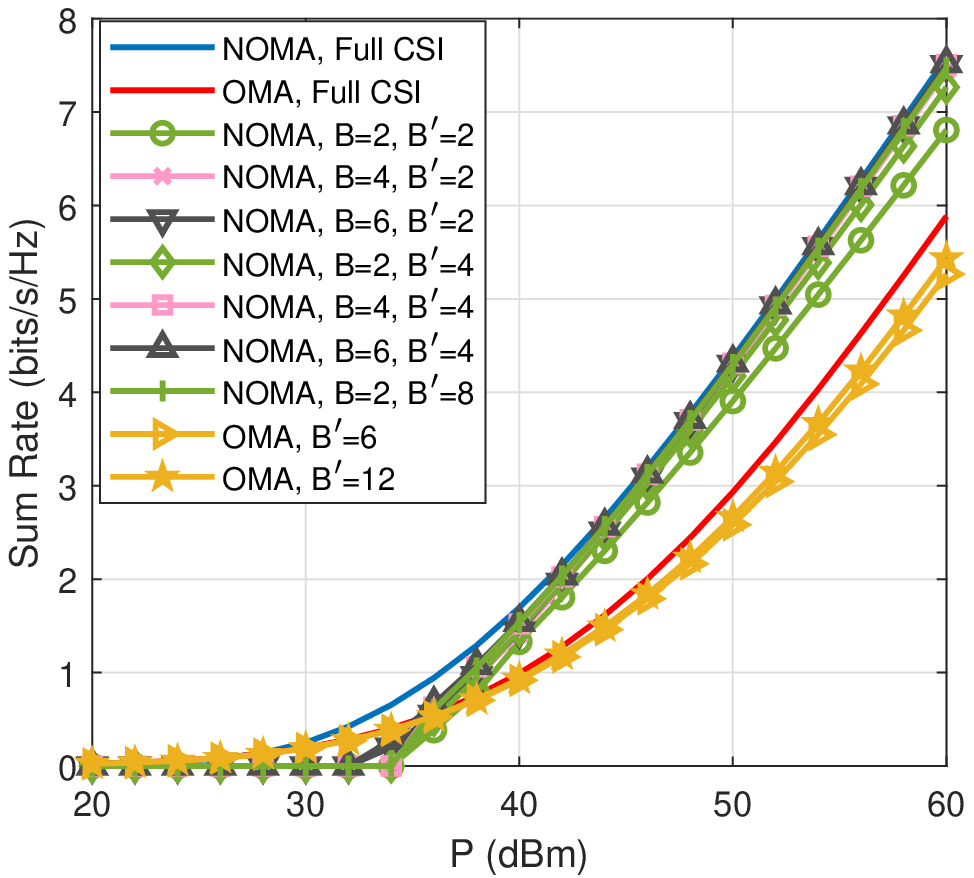}}
         \caption{}
    \end{subfigure}
    \caption{Performance of the sum rate versus the transmit power $P$ for (a) $\zeta_1 = 10^{-5}$ and $\zeta_2 = 0.95$ and (b) $\zeta_1 = 0.5\times 10^{-5}$ and $\zeta_2 = 0.95$.}
   \label{fig:sum_rate}
   \vspace*{-0.6cm}
\end{figure}

We present the simulation results for the sum rate performance versus the total transmit power $P$ for various feedback bits in Fig.~\ref{fig:sum_rate}. The transmit power depends on the users' path loss such that the power should compensate for the propagation loss. 
Simulation is conducted for the full CSI RIS-NOMA, full CSI RIS-OMA, limited feedback RIS-NOMA, and limited feedback RIS-OMA. To study the impact of $\delta$ where $\delta = \zeta_12^{-\zeta_2 B}$, we set $\zeta_2 = 0.95$ and consider two different values for $\zeta_1$. We set $\zeta_1 = 10^{-5}$ and $0.5\times 10^{-5}$ in Figs.~\ref{fig:sum_rate}(a) and \ref{fig:sum_rate}(b), respectively. The full CSI RIS-NOMA achieves the highest sum rate. When $B$ and $B^\prime$ increase, the limited feedback RIS-NOMA's sum rate and the limited feedback RIS-OMA's sum rate approach those of the full CSI RIS-NOMA and the full CSI RIS-OMA, respectively. This is consistent with our findings in Theorem~\ref{theorem_1}. For instance, $B=6$ and $B^\prime=4$, the sum rate is almost the same as that of the full CSI.   
Further, as we decrease the resolution of the quantizer, the length of the region in which we quantize the channel gain to 0 enlarges. Adopting a zero channel gain results in a zero sum rate as indicated in the low power portion of Fig.~\ref{fig:sum_rate}. However, in RIS-OMA, we do not impose any minimum rate constraint. This causes the RIS-OMA's sum rate to be equal or slightly higher than the RIS-NOMA's sum rate at low power levels although there is no guarantee for the minimum sum rate. 

We also observe that for $B=2$ where $B=B_1=B_2$ and $B^\prime=2$, i.e., a total of 6 feedback bits, the limited feedback RIS-NOMA's sum rate shows different behavior compared to the limited feedback RIS-OMA with $B^\prime=6$. At low transmit powers, the limited feedback RIS-OMA's sum rate is better than that of the limited feedback RIS-NOMA for both $\delta$ values. As we increase the power, the limited feedback RIS-NOMA improves the sum rate in comparison to the limited feedback RIS-OMA. However, at very high transmit power levels, the limited feedback RIS-NOMA's slope is smaller than that of the limited feedback RIS-OMA, as shown in Fig.~\ref{fig:sum_rate}(a). 
When we select a smaller $\delta$, as in Fig.~\ref{fig:sum_rate}(b), for any  given $B$ and $B^\prime$, the limited feedback RIS-NOMA's sum rate becomes higher than that of the limited feedback RIS-OMA. 
Another important observation is the impact of allocating $B$ and $B^\prime$ on the limited feedback RIS-NOMA's sum rate. For instance, let us assume the total number of feedback bits is 12. When $B=4$ and $B^\prime=4$, in Fig.~\ref{fig:sum_rate}(a), the sum rate is higher than that of $B=2$ and $B^\prime=8$. Whereas, given the same $B$ and $B^\prime$, in Fig.~\ref{fig:sum_rate}(b), these two schemes achieve almost the same sum rate.  

\begin{figure}
\centering
      \renewcommand{\baselinestretch}{1}
  \includegraphics[scale=0.43]{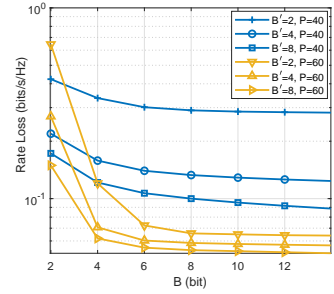}
\caption{The average rate loss versus the number of the feedback bits.}
\label{fig:rate_loss_vs_bits}
\vspace*{-0.6cm}
\end{figure}

Fig.~\ref{fig:rate_loss_vs_bits} compares the limited feedback RIS-NOMA's average rate loss for $\zeta_1 = 0.5\times 10^{-5}$ and $\zeta_2 = 0.95$. As the feedback bits and the power increase, the rate loss reduces. When $B^\prime$ is fixed, by increasing $B$, the rate loss at power $60$ dBm reduces faster than at power $40$ dBm. In fact, increasing the power can compensate for the quantization error. 

\section{Conclusion}\label{sec:conclusion}
In this paper, we studied an FDD-based RIS-NOMA system with a limited feedback channel.  We used a RVQ to quantize the RIS-aided channel vector. Also, we considered a uniform quantizer for the channel gains. We then analyzed the rate loss of the strong user under the accurate and inaccurate user ordering conditions. Since the BS receives the quantized channel gains, inaccurate user ordering can happen often. We derived the rate loss resulted from quantization and showed that the rate loss essentially depends on the number of feedback bits, $B$ and $B^\prime$. From the simulations, we observed that the parameters $B$ and $B^\prime$ affect the sum-rate, the rate loss, and the probability that NOMA is not useful. As the number of feedback bits increases, the quantized RIS-NOMA's  performance approaches that of the full CSI RIS-NOMA.

\appendices
\section{Proof of Lemma~\ref{lemma_1}}\label{Appendix_A}
We divide this super region into Regions II.A-II.D, as shown in Fig.~\ref{fig:graphic}. Due to space limitations, we only provide the detailed calculation of the upper bound and constants for Region II.A. The upper bound and constants for other regions can be found similarly. 

Region II.A: In this region, $H_2 \leq H_1$ which means User~1 is the strong user. It is clear that the output of the RVQ results in  $H_{2,Q} \leq H_2 \leq H_1$ and the uniform quantizer results in $q(H_{2,Q}) \leq q(H_1)$. Thus, the user ordering is accurate and the BS recognizes User~1 as the strong user. The average normalized SNR loss is upper bounded as
\begin{align}\label{eq_XIIA_new}
 \Scale[0.9]{\mathbb{E}[\Delta X_\text{II.A}] \leq  C_6\sqrt{\paren*{2^B-1}\delta} + C_9\delta\sqrt{\paren*{2^B-1}\delta}.}
\end{align}
The proof of \eqref{eq_XIIA_new} and the values of $C_6$ and $C_9$ are provided in Appendix~\ref{Appendix_B}. 

Region II.B: In this region, $H_1 \leq H_2$ indicates User~2 is stronger than User~1. It is possible that the RVQ leads to 
$H_{2,Q} \leq H_1 \leq H_2$. Obviously, the uniform quantizer results in $q(H_1) \geq q(H_{2,Q})$. Thus, the BS recognizes User~1 as the strong user which is not an accurate user ordering. The average normalized SNR loss is expressed as
\begin{align}\label{eq_XIIB_new}
     \Scale[0.9]{\mathbb{E}[\Delta X_\text{II.B}]   \leq C_7\sqrt{\paren*{2^B-1}\delta},}
\end{align}
It should be mentioned that $C_7$ is a constant and bounded since $\Scale[0.9]{\mathbb{E}\left[\frac{1-\eta}{\sqrt{\eta}}\right]}$ is bounded. Also, $\Scale[0.9]{\mathbb{E}\left[\frac{1-\eta}{\sqrt{\eta}}\right]}$ can be approximated as $\Scale[0.9]{\frac{\mathcal{B}\paren*{r_1-\frac{1}{2}, r_2+1}}{\mathcal{B}\paren*{r_1, r_2}}}$ using~\eqref{eq:eta_beta}. 

Region II.C: In this region, $H_1 \leq H_2$ and using the RVQ results in
$H_1 < H_{2,Q} \leq H_2$. For $H_{2,Q} - H_1 < \delta$, the uniform quantizer leads to $q(H_1) = q(H_{2,Q})$. In such a region, the user ordering is inaccurate and the BS picks User~1 as the strong user. Hence, the average normalized SNR loss is obtained as
\begin{align}\label{eq_XIIC_new}
     \Scale[0.9]{\mathbb{E}[\Delta X_\text{II.C}] \leq C_{10}\delta\sqrt{\paren*{2^B-1}\delta} + C_5 \sqrt{2\delta},}
\end{align}
Region II.D: In this region, $H_1 \leq H_2$ and the RVQ leads to $H_1 < H_{2,Q} \leq H_2$. If the uniform quantizer results in $q(H_1) < q(H_{2,Q})$, the user ordering is accurate and the BS selects User~2 as the strong user.  Then, the average normalized SNR loss is bounded by
\begin{align}\label{eq_XIID_new}
    \Scale[0.9]{\mathbb{E}[\Delta X_\text{II.D}] \leq C_8\sqrt{\paren*{2^B-1}\delta} + C_5\delta\sqrt{\paren*{2^B-1}\delta} \text{E}_1\paren*{\frac{\delta}{L_1}},}
\end{align}
The constant $C_8$ is bounded, and using~\eqref{eq:eta_beta} we obtain the approximated value of $\mathbb{E}\left[(1-\eta)\sqrt{\eta}\right]$ as $\frac{\mathcal{B}\paren*{r_1+\frac{1}{2}, r_2+1}}{\mathcal{B}\paren*{r_1, r_2}}$ which is finite.
\section{Proof of~\eqref{eq_XIIA_new}}\label{Appendix_B}
In Region II.A, the normalized SNR loss is obtained as
\begin{align}\label{eq_DXIIA}
    \Scale[0.9]{\Delta X_\text{II.A}} & \Scale[0.9]{= \frac{PH_1H_2 - \epsilon H_1}{PH_2\paren*{1 + \epsilon}} - \frac{PH_1q\paren*{H_{2,Q}} - \epsilon H_1}{Pq\paren*{H_{2,Q}}\paren*{1 + \epsilon}}}  \nonumber \\
    & \Scale[0.9]{\overset{(a)}{\leq} \frac{PH_1H_2 - \epsilon H_1 - PH_1\paren*{H_{2,Q} - \delta} + \epsilon H_1}{PH_2\paren*{1 + \epsilon}}\overset{(b)}{\leq} \frac{H_1H_2 - H_1 H_{2,Q} + \delta H_1}{H_2}} \nonumber \\
    & \Scale[0.9]{\overset{(c)}{=} \frac{H_1 H_2 - \eta H_1 H_2 + \delta H_1}{H_2 } = \paren*{1 - \eta}H_1 + \delta \frac{H_1}{H_2}.}
\end{align}
The inequality~(a) follows from the fact that for any $\delta \leq q(H_{2,Q}) < \paren*{2^B-1}\delta$, the inequalities $q(H_{2,Q}) \geq H_{2,Q}-\delta$ and $H_{2,Q} \geq q(H_{2,Q})$ hold. The inequality~(b) is due to $\epsilon \geq 0$. Further, (c) is true because $H_{2,Q} = \eta H_2$. For a given constant $\eta$, the expectation of $\Delta X_\text{II.A}$  over the super region defined by $\frac{\delta}{\eta} \leq H_2 < \frac{\paren*{2^B-1}\delta}{\eta}$ and $H_2 \leq H_1$ is given as 
\begin{align}\label{eq:DXIIAeta}
    & \Scale[0.9]{\mathbb{E}[\Delta X_\text{II.A} | \eta] \leq \paren*{1 - \eta} \underbrace{\int_{\frac{\delta}{\eta}}^{\frac{\paren*{2^B-1}\delta}{\eta}} \int_{H_2}^{\infty} H_1f_{H_1}\paren*{H_1}f_{H_2}\paren*{H_2} dH_1dH_2}_{\equiv I_1}} \nonumber \\ 
    & \quad \Scale[0.9]{+ \delta \underbrace{\int_{\frac{\delta}{\eta}}^{\frac{\paren*{2^B-1}\delta}{\eta}} \int_{H_2}^{\infty}  \frac{H_1}{H_2}f_{H_1}\paren*{H_1}f_{H_2}\paren*{H_2}dH_1dH_2}_{\equiv I_2}}. 
\end{align}
Next, we compute the integrals $I_1$ and $I_2$. For $I_1$, we have
\begin{align} \label{eq_I1}
    \Scale[0.9]{I_1} & \Scale[0.9]{= \int_{\frac{\delta}{\eta}}^{\frac{\paren*{2^B-1}\delta}{\eta}}  e^{-\frac{H_2}{L_1}}\paren*{H_2 + L_1}f_{H_2}\paren*{H_2}dH_2} \nonumber \\
    & \Scale[0.9]{\leq \frac{C_1}{2}\underbrace{\int_{\frac{\delta}{\eta}}^{\frac{\paren*{2^B-1}\delta}{\eta}}
    \paren*{\frac{\sqrt{H_2}}{C_2}}^{\frac{3N}{2}}e^{-\paren*{\frac{H_2}{L_1} + \frac{2\sqrt{H_2}}{C_2}}}dH_2}_{\equiv I_{1,1}}} \nonumber \\
    & \quad \Scale[0.9]{+ \frac{C_1L_1}{2C_2^2}\underbrace{\int_{\frac{\delta}{\eta}}^{\frac{\paren*{2^B-1}\delta}{\eta}}\paren*{\frac{\sqrt{H_2}}{C_2}}^{\frac{3N-4}{2}}e^{-\paren*{\frac{H_2}{L_1} + \frac{2\sqrt{H_2}}{C_2}}}dH_2}_{\equiv I_{1,2}}.}
\end{align}
The integral $I_{1,1}$ is obtained as
\begin{align} \label{eq_I11_new}
    & \Scale[0.9]{I_{1,1}} \Scale[0.9]{\leq \int_0^{\frac{\paren*{2^B-1}\delta}{\eta}}
    \paren*{\frac{\sqrt{H_2}}{C_2}}^{\frac{3N}{2}}e^{-\paren*{\frac{H_2}{L_1} + \frac{2\sqrt{H_2}}{C_2}}}dH_2} \nonumber \\
    & \Scale[0.86]{\overset{(a)}{\leq} \int_0^{\frac{\paren*{2^B-1}\delta}{\eta}}
    \paren*{\frac{\sqrt{H_2}}{C_2}}^{\frac{3N}{2}}e^{-\frac{2\sqrt{H_2}}{C_2}}dH_2 \overset{(b)}{\leq} \frac{C_2}{2^{\frac{3N}{2}}}\paren*{\frac{3N+2}{2}}!\sqrt{\frac{\paren*{2^B-1}\delta}{\eta}}.} 
\end{align}
The inequality $(a)$ follows from the fact that  $e^{-\frac{H_2}{L_1}} \leq 1$. The inequality~$(b)$ is due to the definition of the lower incomplete gamma function  and using the upper bound $\Scale[0.9]{\gamma(n, x) \leq (n-1)! x}$.

Likewise, we obtain an upper bound on $I_{1,2}$ as
\begin{align}\label{eq_I12new}
    \Scale[0.9]{I_{1,2} 
    \leq \frac{C_2}{2^{\frac{3N-4}{2}}}\paren*{\frac{3N-2}{2}}!\sqrt{\frac{\paren*{2^B-1}\delta}{\eta}}.} 
\end{align}
Substituting~\eqref{eq_I11_new} and~\eqref{eq_I12new} into~\eqref{eq_I1} gives
\begin{align}  \label{eq_I1final_new}
    \Scale[0.9]{I_1}  &\Scale[0.9]{\leq \frac{C_1C_2}{2^{\frac{3N+2}{2}}}\paren*{\frac{3N+2}{2}}!\sqrt{\frac{\paren*{2^B-1}\delta}{\eta}}} \nonumber \\
    & \quad \Scale[0.9]{+ \frac{C_1L_1}{2^{\frac{3N-2}{2}}C_2}\paren*{\frac{3N-2}{2}}!\sqrt{\frac{\paren*{2^B-1}\delta}{\eta}} = C_6^{'}\sqrt{\frac{\paren*{2^B-1}\delta}{\eta}},}
\end{align}
Also, we calculate $I_2$ as
\begin{align}\label{eq_I2upper}
    \Scale[0.9]{I_2 = \int_{\frac{\delta}{\eta}}^{\frac{\paren*{2^B-1}\delta}{\eta}} e^{-\frac{H_2}{L_1}}   \paren*{1 + \frac{L_1}{H_2}}f_{H_2}\paren*{H_2}dH_2 \leq C_9^{'}\sqrt{\frac{\paren*{2^B-1}\delta}{\eta}},} 
\end{align}
Thus, the upper bound on $\mathbb{E}[\Delta X_\text{II.A} | \eta]$ is found as 
\begin{align}
    \Scale[0.9]{\mathbb{E}[\Delta X_\text{II.A} | \eta] \leq \paren*{1-\eta}I_1 + \delta I_2,}
\end{align}
where the upper bounds on $I_1$ and $I_2$ are derived in~\eqref{eq_I1final_new} and \eqref{eq_I2upper}, respectively. In order to calculate $\mathbb{E}[\Delta X_\text{II.A}]$, we have
\begin{align}\label{eq_XIIA}
    \Scale[0.9]{\mathbb{E}[\Delta X_\text{II.A}]} & \leq  \Scale[0.9]{\int_0^1 (1-\eta)I_1f_\eta(\eta)d\eta + \delta\int_0^1 I_2f_\eta(\eta)d\eta} \nonumber \\ 
    & \leq \Scale[0.9]{C_6^{'}\sqrt{\paren*{2^B-1}\delta}\int_0^1 \frac{1-\eta}{\sqrt{\eta}}f_\eta(\eta)d\eta} \nonumber \\
    &  \quad \Scale[0.9]{ + C_9^{'}\delta\sqrt{\paren*{2^B-1}\delta} \int_0^1\frac{1}{\sqrt{\eta}}f_\eta(\eta)d\eta} \nonumber \\ 
    & \Scale[0.9]{= C_6^{'} \mathbb{E}\left[ \frac{1-\eta}{\sqrt{\eta}}\right] \sqrt{\paren*{2^B-1}\delta} + C_9^{'}  \mathbb{E}\left[ \frac{1}{\sqrt{\eta}}\right] \delta\sqrt{\paren*{2^B-1}\delta}} \nonumber \\ 
    & \leq \Scale[0.9]{C_6 \sqrt{\paren*{2^B-1}\delta} + C_9 \delta\sqrt{\paren*{2^B-1}\delta},}
\end{align}
Note that $\Scale[0.9]{\mathbb{E}\left[ \frac{1-\eta}{\sqrt{\eta}}\right]}$ and $ \Scale[0.9]{\mathbb{E}\left[\frac{1}{\sqrt{\eta}}\right]}$ are finite and for the approximation in \eqref{eq:eta_beta} are, respectively, equal to $\Scale[0.9]{\frac{\mathcal{B}\paren*{r_1-\frac{1}{2}, r_2+1}}{\mathcal{B}\paren*{r_1, r_2}}}$ and $\Scale[0.9]{\frac{\mathcal{B}\paren*{r_1-\frac{1}{2},r_2}}{\mathcal{B}(r_1, r_2)}}$.

\section{Proof of Lemma~\ref{lemma_2}} \label{Appendix_C}
We divide Super Region III into Regions III.A-III.C, as shown in Fig.~\ref{fig:graphic}.  

Region III.A: In this region, User~1 has a better channel compared to User~2, i.e., $H_2 \leq H_1$. If the RVQ results in
$\paren*{2^B-1}\delta \leq H_{2,Q} \leq H_2 \leq H_1$, then uniform quantizer's outcome will be $q(H_1) = q(H_{2,Q})=\paren*{2^B-1}\delta$ and the ordering is accurate. In this region, the average normalized SNR loss is bounded by
\begin{align}\label{DXIIIA_new}
    \Scale[0.9]{\mathbb{E}[\Delta X_\text{III.A}] \leq} & \Scale[0.8]{e^{-2\frac{\sqrt{\paren*{2^B-1}\delta}}{C_2}}\Bigg[C_{11} \paren*{1+\paren*{ 2\sqrt{\frac{\paren*{2^B-1}\delta}{ C_2^2}}}^{\frac{3N+2}{2}}}} \nonumber \\ 
    & \quad \Scale[0.9]{+ C_{12}\paren*{1 + \paren*{2\sqrt{\frac{\paren*{2^B-1}\delta}{ C_2^2}}}^{\frac{3N-2}{2}}}}\Bigg], 
\end{align}

Region III.B: In this region, User~2 has a better channel and $q(H_1) = \paren*{2^B-1}\delta$, i.e., $\paren*{2^B-1}\delta \leq H_1 \leq H_2$. In such a region, the RVQ will result in either
$\paren*{2^B-1}\delta \leq H_1 \leq H_{2,Q}$ or $\paren*{2^B-1}\delta \leq H_{2,Q} \leq H_1$. Also, the uniform quantizer will provide $q(H_1) = q(H_{2,Q})=\paren*{2^B-1}\delta$. Hence, the quantizers will inaccurately change the users' order. The average normalized SNR loss for this region is bounded as
\begin{align}\label{DXIIIB_new}
    \Scale[0.9]{\mathbb{E}[\Delta X_\text{III.B}]\leq C_{11}e^{-2\frac{\sqrt{\paren*{2^B-1}\delta}}{C_2}}\paren*{1 + \paren*{2\sqrt{\frac{\paren*{2^B-1}\delta}{C_2^2}}}^{\frac{3N+2}{2}}}}.
\end{align}

Region III.C: In this region, User~1's  channel gain is lower than that of User~2 such that $\delta \leq q(H_1)<q(H_{2,Q}) = \paren*{2^B-1}\delta$. Thus, User~2 is the strong user and  the ordering is accurate. We obtain the expectation of the normalized SNR loss as
\begin{align}\label{eq_DXIIICeta_new}
    \Scale[0.85]{\mathbb{E}[\Delta X_\text{III.C}] \leq  2\frac{C_{11}}{L_1}\paren*{2^B-1}\delta e^{-2\frac{\sqrt{\paren*{2^B-1}\delta}}{C_2}} \paren*{1+\paren*{2\sqrt{\frac{\paren*{2^B-1}\delta}{C_2^2}}}^{\frac{3N+2}{2}}}}.
\end{align}

\bibliographystyle{IEEEtran}
\bibliography{IEEEabrv,references}
\end{document}